\documentclass{article}
\pdfoutput=1
\usepackage[whileod]{clr-alg}
\usepackage[x11names]{xcolor}
\colorlet{shadecolor}{LavenderBlush3!20}
\usepackage{tikz}
\usetikzlibrary{arrows,automata,decorations.pathmorphing,positioning}
\tikzset{initial text=$$} 

\newtheorem{theorem}{Theorem}
\newtheorem{proposition}[theorem]{Proposition}
\newtheorem{corollary}[theorem]{Corollary}

\newenvironment{proof}{\noindent\textbf{Proof.} }{\hfill\rule{2mm}{2mm}\medskip}

\def\dd{\mathinner{\ldotp\ldotp}}   
\newcommand{\strl}{\mathrm{\;<\hspace{-1.5mm}<\;}} 
\newcommand{\per}{\mathit{per}}
\newcommand{\leftchild}{\mathit{left}}
\newcommand{\rightchild}{\mathit{right}}
\newcommand{\period}{\mathit{period}}
\newcommand{\lyns}{\mathit{LynS}} 
\newcommand{\lyn}{\mathit{Lyn}} 
\newcommand{\llynt}{\mathcal{L}} 
\newcommand{\roo}{\mathit{root}}
\newcommand{\ran}{\mathit{rank}}
\newcommand{\psp}{\mathit{psp}}

\title{Left Lyndon tree construction\footnote{Revision and extension of a contribution to Prague Stringology Conference 2020 \cite{BadkobehC20}}}
\author{%
Golnaz Badkobeh%
\thanks{Goldsmiths, University of London, New Cross, London SE14 6NW, UK.
\texttt{g.badkobeh@gold.ac.uk}}
\ \&
Maxime Crochemore%
\thanks{King's College London, Informatics, 30 Aldwych, London WC2B 4BG, UK,
 and Universit\'e Gustave Eiffel, 77454 Marne-la-Vall\'ee, France.
 \texttt{Maxime.Crochemore@kcl.ac.uk}}
}

\begin{document}
\maketitle
\begin{abstract}
We extend the left-to-right Lyndon factorisation of a word to the left Lyndon tree construction of a Lyndon word. It yields an algorithm to sort the prefixes of a Lyndon word according to the infinite ordering defined by Dolce et al. (2019). A straightforward variant computes the left Lyndon forest of a word. All algorithms run in linear time on a general alphabet, that is, in the letter-comparison model.
\end{abstract}

\section{Lyndon words}\label{sect:intro}

In this article we consider algorithmic questions related to Lyndon words.
Introduced in the field of combinatorics by Lyndon (see \cite{Lyndon54}) and used in algebra, these words have shown their usefulness for designing efficient algorithms on words. The notion of Lyndon tree associated with the decomposition of a Lyndon word has, for example, been used by Bannai et al. \cite{BannaiIINTT17} to solve a conjecture of Kolpakov and Kucherov \cite{KolpakovK99} on the maximal number of runs (maximal periodicities) in words, following a result in \cite{CrochemoreIKRRW12}.

The key result in \cite{BannaiIINTT17} is that every run in a word $y$ contains as a factor a Lyndon root (according to the alphabet ordering or its inverse) that corresponds to a node of the associated Lyndon tree. Since the Lyndon tree has a linear number of nodes according to the length of $y$, browsing all its nodes leads to a linear-time algorithm in order to report all the runs occurring in $y$. However, the time complexity of this technique also depends on the time it takes to build the tree and to extend a potential run root to an actual run.

Here we consider the left Lyndon tree of a Lyndon word $y$. This binary tree has a single node if $y$ is reduced to a single letter, otherwise its structure parallels recursively the left standard factorisation (see Viennot \cite{Viennot78}) of $y$ as $uv$ where $u$ is the longest proper Lyndon prefix of $y$.

The dual notion of right Lyndon tree of a Lyndon word $y$ (based on the factorisation $y=uv$ where $v$ is the longest proper Lyndon suffix of $y$) is strongly related to the sorted list of suffixes of $y$. Indeed, Hohlweg and Reutenauer \cite{HohlwegR03} showed that the tree is the Cartesian tree built from ranks of suffixes in their lexicografically sorted list (see \cite{CrochemoreR20}). The list corresponds to the standard permutation of suffixes of the word and is the main component of its suffix array (see \cite{ManberM90} or \verb"en.wikipedia.org/wiki/Suffix_array"), one of the major data structures for text indexing.

Inspired by a result of Ufnarovskij \cite{Ufnarovskij11}, Dolce et al. \cite{DolceRR19b} showed that the left Lyndon tree is also a Cartesian tree built from the ranks of prefixes sorted according to an ordering they call the infinite order.

The main result of this article is to show that sorting prefixes of a Lyndon word according to the infinite ordering can be attained in linear time in the letter-comparison model. This produces the prefix standard permutation of the word. The algorithm is based on the Lyndon factorisation of words by Duval \cite{Duval83} and it extends naturally to build the left Lyndon forest of a word. Furthermore, recovery of a word from its prefix standard permutation can be made in linear time.

Recently, Bille et al. \cite{BilleE0GKMR20} designed an algorithm to build the right Lyndon table of a word in linear time on a general alphabet, result from which the right Lyndon tree can be deduced with the same time complexity. The reverse-engineering question on this table is discussed by Nakashima et al. in \cite{NakashimaTIBT19}.

\subsection*{Definitions}
Let $A$ be an alphabet with an ordering $<$ and $A^+$ be the set of non-empty words with the lexicographical ordering induced by $<$. The length of a word $w$ is denoted by $|w|$. We say that $uv$ (formally $(u,v)$) is a non-trivial factorisation of a word $w$ if $uv=w$ and both $u$ and $v$ are non-empty words.

A word is said to be strongly less than a word $v$, denoted by $u\strl v$, if there are words $r$, $s$ and $t$, and letters $a$ and $b$ satisfying $u=ras$, $v=rbt$ and $a<b$. And a word $u$ is smaller than a word $v$, $u<v$, if either $u\strl v$ or $u$ is a proper prefix of $v$.

In addition to the usual lexicographical ordering, the infinite ordering denoted by $\prec$ (see \cite{DolceRR19,DolceRR19b}) is defined by: $u\prec v$ if $u^\infty < v^\infty$ or both $u^\infty = v^\infty$ and $|u| > |v|$. Note that the condition $u^\infty = v^\infty$ implies that $u$ and $v$ are powers of the same word, consequence of Fine and Wilf's Periodicity lemma (see \cite[Proposition 1.3.5]{Lothaire83}).

Let $\mathcal{L}$ be the set of Lyndon words on the alphabet $A$. The next proposition defines Lyndon words that are not reduced to a single letter. Condition in item (i) is the original definition and condition in item (iii) is by Ufnarovskij \cite{Ufnarovskij11}.

\begin{proposition}\label{prop-1}
Any of the following equivalent conditions define a Lyndon word $w$, $|w|>1$:
(i) $w < vu$, for any non-trivial factorisation  $uv$ of $w$, (ii) $w < v$, for any proper non-empty suffix $v$ of $w$, (iii) $u^\infty < w^\infty$, for any proper non-empty prefix $u$ of~$w$.
\end{proposition}

\section{Lyndon suffix table}\label{sect:lyns}

Algorithms presented in the article strongly use the notion of Lyndon suffix table of a word, which is denoted by $\lyns$. The table $\lyns$ (more accurately  $\lyns_y$) of a word $y$ is defined, for each position $j$ on $y$, by
$$\lyns[j]=\max\{|w| \mid w \mbox{ longest Lyndon suffix of } y[0\dd j]\}.$$
For $y = \texttt{babbababbaabb}$ on the alphabet of constant letters $\{\texttt{a}, \texttt{b}, \dots\}$ ordered as usual $\texttt{a} < \texttt{b} < \cdots$, the $\lyns$ table is as follows:

\smallskip\noindent
\begin{tabular}{@{}l@{\quad}*{13}{p{9pt}}l@{}}
$j$		&0&1&2&3&4&5&6&7&8&9&10&11&12 \\
\hline
$y[j]$	&\tt b&\tt a&\tt b&\tt b&\tt a&\tt b&\tt a&\tt b&\tt b&\tt a&\tt a&\tt b&\tt b \\
$\lyns[j]$	&1&1&2&3&1&2&1&2&5&1&1&3&4
\end{tabular}

\medskip\noindent
Table $\lyns$ is the dual notion of the Lyndon table of $y$ (also called Lyndon array) $l$ in \cite{BannaiIINTT17}, $\mathcal{L}$ in \cite{FranekL19} or $\lyn$ in \cite{CrochemoreR20,CLR20cup}, used to detect maximal periodicities (runs) in words: $\lyn[j]$ is the maximal length of Lyndon prefixes of $y|j\dd |y|-1]$.

The computation of $\lyns$ is a mere extension of the algorithm for testing if a word is the prefix of a Lyndon word. It includes the key point of the factorisation algorithm in \cite{Duval83} and is recalled first as Algorithm \Algo{LyndonWordPrefix} that tests if its input is a prefix of a Lyndon word and that works online on its input. Note that it is a Lyndon word if its final period equals its length.

\medskip\noindent
\begin{algo}{LyndonWordPrefix}{y \textrm{ non-empty word of length } n}
  \SET{(\per,i)}{(1,0)}
  \DOFORI{j}{1}{n-1}
    \IF{y[j] > y[i]} \RCOM{6}{$y[i] = y[j-\per]$} \label{alg1-line3}
      \SET{(\per,i)}{(j+1,0)} \label{alg1-line4}
    \ELSEIF{y[j] < y[i]}
      \RETURN{\False}
    \ELSE
      \SET{i}{i+1 \bmod \per}
    \FI
  \OD
  \RETURN{\True}
\end{algo}

\noindent
{\begin{picture}(300,35)(0,0)
\put(  0,18){$y$}
\put( 10,15){\framebox(300,10){}}
\put( 10,15){\framebox(100,10){$x$}}
\put(110,15){\framebox(100,10){$x$}}
\put(210,15){\framebox(30,10){$z$}}
\put( 10,29){\makebox(7,7){$0$}}
\put( 40,29){\makebox(10,7){$i$}}
\put(240,28){\makebox(10,7){$j$}}
\put( 10, 0){\framebox(30,10){$z$}}
\put(110,11){\vector(1,0){100}}
\put(210,11){\vector(-1,0){100}}
\put(110, 0){\makebox(100,7){$\per$}} 
\put( 10,10){\dashbox{1}(0,5)}
\put( 40,10){\dashbox{1}(0,15)}
\end{picture}}

\medskip
The key feature of the method stands in lines \ref{alg1-line3}-\ref{alg1-line4} of the algorithm and is illustrated on the above picture. If $y[j]>y[i]=y[j-\per]$, not only the periodicity $\per$ of $y[0\dd j-1]$ breaks but $y[0\dd j]$ is a Lyndon word with period $j+1$. This results from the following known properties (see \cite{Lothaire83}).

\begin{proposition}\label{prop-2}
(i) Let $z$ be a word and $a$ a letter for which $za$ is a prefix of a Lyndon word and let $b$ be a letter with $a<b$. Then $zb$ is a Lyndon word.\\
(ii) Let $u$ and $v$ be two Lyndon words with $u<v$. Then $uv$ is a Lyndon word.
\end{proposition}

Algorithm \Algo{LyndonSuffixT} below computes the Lyndon suffix table of a Lyndon word. (It is extended in Section~\ref{sect:forest} to compute the same table of a non-empty word.) The algorithm results from a minor modification of Algorithm \Algo{LyndonWordPrefix} and can be easily enhanced to compute also the period of all non-empty prefixes of the input.

\medskip\noindent
\begin{algo}{LyndonSuffixT}{y \textrm{ Lyndon word of length } n}
  \SET{\lyns[0]}{1}
  \SET{(\per,i)}{(1,0)}
  \DOFORI{j}{1}{n-1}
    \IF{y[j] \neq y[i]} \RCOM{6}{$y[j] > y[i] = y[j-\per]$}
      \SET{\lyns[j]}{j+1}
      \SET{(\per,i)}{(j+1,0)}
    \ELSE
      \SET{\lyns[j]}{\lyns[i]}
      \SET{i}{i+1 \bmod \per}
    \FI
  \OD
  \RETURN{\lyns}
\end{algo}

\begin{proposition}\label{prop-3}
Algorithm \Algo{LyndonSuffixT} computes the Lyndon suffix table of a Lyndon word of length $n$ in time $O(n)$ in the letter-comparison model.
\end{proposition}

Given $y=\texttt{ababbababbabac}$, the corresponding $\lyns$ table and period table are as follows:

\medskip\noindent
\begin{tabular}{@{}l@{\quad}*{14}{p{9pt}}l@{}}
$j$		&0&1&2&3&4&5&6&7&8&9&10&11&12&13 \\
\hline
$y[j]$	&\tt a&\tt b&\tt a&\tt b&\tt b&\tt a&\tt b&\tt a&\tt b&\tt b&\tt a&\tt b&\tt a&\tt c\\
$\lyns[j]$	&1&2&1&2&5&1&2&1&2&5&1&2&1&14\\
$\period[j]$ &1&2&2&2&5&5&5&5&5&5&5&5&5&14
\end{tabular}

\section{Left Lyndon tree construction}\label{sect:llyn}

The left Lyndon tree $\llynt(y)$ of a Lyndon word $y$ represents recursively the left standard factorisation of $y$. It is a binary tree whose leaves are positions on the word and internal nodes correspond to concatenations of two Lyndon factors of the word, and as such can be viewed as interpositions. Precisely, $\llynt(y)=(p)$ if $|y|=1$ else it is $(p,\llynt(u),\llynt(v))$ where the node $p\in\{|y|\dd 2|y|-2\}$ is an integer and $uv$ is the left standard factorisation of $y$, that is, $u$ is the longest proper Lyndon prefix of $y$ ($v$ is then a Lyndon word).

In the next algorithm, subtrees of $\llynt(y)$ are handled from positions on $y$ as follows. The subtree associated with position $j$ is $\llynt(y[i\dd j])$ where $j-i+1=\lyns[j]$ and its root is $\roo[j]$. Thus, position $j$ on $y$ is the rightmost leaf of the subtree and $\lyns[j]$ is its tree width. Besides, the left child of an internal node $q$ is $\leftchild(q)$ and its right child is $\rightchild(q)$.

It is known that $y$, as a Lyndon word with $|y|>1$, is of the form $x^kzb$ where $x$ is a Lyndon word of length $\per=\period(x^kz)$, $k>0$, $z$ is a proper prefix of $x$ and $b$ is a letter greater than letter $a$ following $z$ in $x$ ($za$ is a prefix of $x$) \cite{Duval83}.

The construction of $\llynt(y)$ is achieved with the help of the table $\lyns$ of $y$. It is done by processing $y$ from left to right building first $\llynt(x)$ and reproducing that tree or part of it up to $z$. The picture displays the subtrees built for the word $(\texttt{ababb})^2\texttt{aba}$.

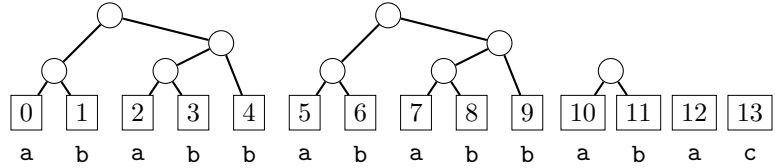
\tikzstyle{intern} = [circle, draw, thin]
\tikzstyle{leaf} = [draw,very thin]
\bigskip
\begin{tikzpicture}[scale=0.37,node distance=1.5cm, auto,>=latex', thick]
\path[->]
 node[leaf] (p0) at (0,1.5) {$0$}
 node[leaf] (p1) at (2,1.5) {$1$}
 node[leaf] (p2) at (4,1.5) {$2$}
 node[leaf] (p3) at (6,1.5) {$3$}
 node[leaf] (p4) at (8,1.5) {$4$}
 node[leaf] (p5) at (10,1.5) {$5$}
 node[leaf] (p6) at (12,1.5) {$6$}
 node[leaf] (p7) at (14,1.5) {$7$}
 node[leaf] (p8) at (16,1.5) {$8$}
 node[leaf] (p9) at (18,1.5) {$9$}
 node[leaf] (p10) at (20,1.5) {$10$}
 node[leaf] (p11) at (22,1.5) {$11$}
 node[leaf] (p12) at (24,1.5) {$12$}
 node[leaf] (p13) at (26,1.5) {$13$}
;
\path[->]
 node (q0) at (0,0) {\tt a}
 node (q1) at (2,0) {\tt b}
 node (q2) at (4,0) {\tt a}
 node (q3) at (6,0) {\tt b}
 node (q4) at (8,0) {\tt b}
 node (q5) at (10,0) {\tt a}
 node (q6) at (12,0) {\tt b}
 node (q7) at (14,0) {\tt a}
 node (q8) at (16,0) {\tt b}
 node (q9) at (18,0) {\tt b}
 node (q10) at (20,0) {\tt a}
 node (q11) at (22,0) {\tt b}
 node (q12) at (24,0) {\tt a}
 node (q13) at (26,0) {\tt c}
;
\path[-]
 node[intern] (p14) at (1,3) {}
 node[intern] (p15) at (5,3) {}
 node[intern] (p16) at (7,4) {}
 node[intern] (p17) at (3,5) {}
 node[intern] (p18) at (11,3) {}
 node[intern] (p19) at (15,3) {}
 node[intern] (p20) at (17,4) {}
 node[intern] (p21) at (13,5) {}
 node[intern] (p22) at (21,3) {}
 (p14) edge node {} (p0)
 (p14) edge node {} (p1)
 (p15) edge node {} (p2)
 (p15) edge node {} (p3)
 (p16) edge node {} (p4)
 (p16) edge node {} (p15)
 (p17) edge node {} (p14)
 (p17) edge node {} (p16)
 (p18) edge node {} (p5)
 (p18) edge node {} (p6)
 (p19) edge node {} (p7)
 (p19) edge node {} (p8)
 (p20) edge node {} (p9)
 (p20) edge node {} (p19)
 (p21) edge node {} (p18)
 (p21) edge node {} (p20)
 (p22) edge node {} (p10)
 (p22) edge node {} (p11)
;
\end{tikzpicture}

\medskip
The main step of the procedure, in addition to computing $\lyns$ identically as in  Algorithm \Algo{LyndonSuffixT} above, is to aggregate partial Lyndon trees when processing the last letter $b$ of $y$, which creates the final tree as a bundle of all subtrees. In fact, this step is also carried out when dealing with $x^kz$ at each position $j$ for which $\lyns[j]>1$. In order to aggregate the subtrees, the second property of Proposition~\ref{prop-2}  is applied iteratively, processing the trees from right to left. An explicit instruction of this step is designed at lines~\ref{alg3-line10}-\ref{alg3-line15} in Algorithm \Algo{LeftLyndonTree} below.

The process of bundling can be viewed as a translation into the tree structure of the proof of the key feature of Algorithm \Algo{LyndonWordPrefix} stated in item (ii) of Proposition~\ref{prop-2}. Even so the latter algorithm deals with this process in constant time using item (i) of the proposition, the iteration of instructions during the bundling does not affect the asymptotic running time of the present algorithm.

\medskip\noindent
\begin{algo}{LeftLyndonTree}{y \textrm{ Lyndon word of length } n}
  \SET{(\lyns[0],\roo[0])}{(1,0)}
  \SET{(\per,i)}{(1,0)}
  \DOFORI{j}{1}{n-1}
    \SET{\roo[j]}{j}
    \IF{y[j] \neq y[i]} \RCOM{6}{$y[j] > y[i] = y[j-\per]$}
      \SET{\lyns[j]}{j+1}
      \SET{(\per,i)}{(j+1,0)}
    \ELSE
      \SET{\lyns[j]}{\lyns[i]} \label{alg3-line8}
      \SET{i}{i+1 \bmod \per}
    \FI
    \SET{(\ell,k)}{(1,j-1)} \label{alg3-line10}
    \DOWHILE{\ell < \lyns[j]}
      \SET{q}{\mbox{new node} \geq n}
      \SET{(\leftchild[q],\rightchild[q])}{(\roo[k],\roo[j])}
      \SET{\roo[j]}{q}
      \SET{(\ell,k)}{(\ell+\lyns[k],k-\lyns[k])} \label{alg3-line15}
    \OD
  \OD
  \RETURN{\roo[n-1]}
\end{algo}

The picture below shows red nodes and links created by the final round of instructions at lines~\ref{alg3-line10}-\ref{alg3-line15} in Algorithm \Algo{LeftLyndonTree}.

\medskip
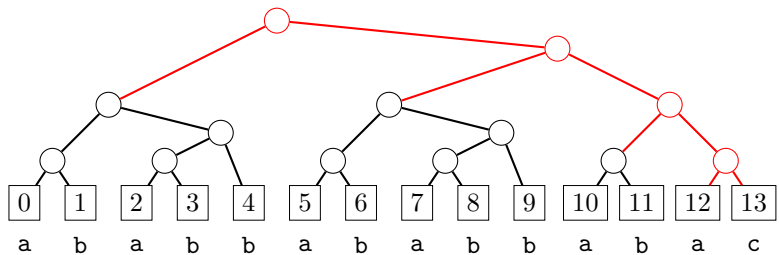
\begin{tikzpicture}[scale=0.373,node distance=1.5cm, auto,>=latex', thick]
\path[->]
 node[leaf] (p0) at (0,1.5) {$0$}
 node[leaf] (p1) at (2,1.5) {$1$}
 node[leaf] (p2) at (4,1.5) {$2$}
 node[leaf] (p3) at (6,1.5) {$3$}
 node[leaf] (p4) at (8,1.5) {$4$}
 node[leaf] (p5) at (10,1.5) {$5$}
 node[leaf] (p6) at (12,1.5) {$6$}
 node[leaf] (p7) at (14,1.5) {$7$}
 node[leaf] (p8) at (16,1.5) {$8$}
 node[leaf] (p9) at (18,1.5) {$9$}
 node[leaf] (p10) at (20,1.5) {$10$}
 node[leaf] (p11) at (22,1.5) {$11$}
 node[leaf] (p12) at (24,1.5) {$12$}
 node[leaf] (p13) at (26,1.5) {$13$}
;
\path[->]
 node (q0) at (0,0) {\tt a}
 node (q1) at (2,0) {\tt b}
 node (q2) at (4,0) {\tt a}
 node (q3) at (6,0) {\tt b}
 node (q4) at (8,0) {\tt b}
 node (q5) at (10,0) {\tt a}
 node (q6) at (12,0) {\tt b}
 node (q7) at (14,0) {\tt a}
 node (q8) at (16,0) {\tt b}
 node (q9) at (18,0) {\tt b}
 node (q10) at (20,0) {\tt a}
 node (q11) at (22,0) {\tt b}
 node (q12) at (24,0) {\tt a}
 node (q13) at (26,0) {\tt c}
;
\path[-]
 node[intern] (p14) at (1,3) {}
 node[intern] (p15) at (5,3) {}
 node[intern] (p16) at (7,4) {}
 node[intern] (p17) at (3,5) {}
 node[intern] (p18) at (11,3) {}
 node[intern] (p19) at (15,3) {}
 node[intern] (p20) at (17,4) {}
 node[intern] (p21) at (13,5) {}
 node[intern] (p22) at (21,3) {}
 node[intern,color=red] (p23) at (25,3) {}
 node[intern,color=red] (p24) at (23,5) {}
 node[intern,color=red] (p25) at (19,7) {}
 node[intern,color=red] (p26) at ( 9,8) {}
 (p14) edge node {} (p0)
 (p14) edge node {} (p1)
 (p15) edge node {} (p2)
 (p15) edge node {} (p3)
 (p16) edge node {} (p4)
 (p16) edge node {} (p15)
 (p17) edge node {} (p14)
 (p17) edge node {} (p16)
 (p18) edge node {} (p5)
 (p18) edge node {} (p6)
 (p19) edge node {} (p7)
 (p19) edge node {} (p8)
 (p20) edge node {} (p9)
 (p20) edge node {} (p19)
 (p21) edge node {} (p18)
 (p21) edge node {} (p20)
 (p22) edge node {} (p10)
 (p22) edge node {} (p11)
 (p23) edge[color=red] node {} (p12)
 (p23) edge[color=red] node {} (p13)
 (p24) edge[color=red] node {} (p22)
 (p24) edge[color=red] node {} (p23)
 (p25) edge[color=red] node {} (p21)
 (p25) edge[color=red] node {} (p24)
 (p26) edge[color=red] node {} (p17)
 (p26) edge[color=red] node {} (p25)
;
\end{tikzpicture}

\begin{proposition}\label{prop-4}
Algorithm \Algo{LeftLyndonTree} builds the left Lyndon tree of a Lyndon word of length $n$ in time $O(n)$ in the letter-comparison model.
\end{proposition}
\begin{proof}
All instructions inside the \For\ loop execute in constant time except the \While\ loop. In addition, since each execution of instructions in the \While\ loop takes constant time and leads to the creation of an internal node of the final tree twinned with the fact that there are eactly $n-1$ such nodes, the total (amortised) running time is $O(n)$.
\end{proof}

\section{Sorting prefixes}\label{sect:sort}

This section shows that Algorithm \Algo{LeftLyndonTree} can be adapted to sort the prefixes of a Lyndon word according to the infinite ordering $\prec$. This is a consequence of Theorem~\ref{theo-6} below.

For the Lyndon word $y$, an internal node $p$ of the left Lyndon tree $\llynt(y)$ is the root of a Lyndon subtree associated with a Lyndon factor $w$ of $y$. This factor is obtained by concatenating two consecutive occurrences of Lyndon factors $u$ and $v$. If the concerned occurrence of $u$ ends at position $j$ on $y$, node $p$ is identified with the prefix of $y$ ending at position $j$. The correspondence between internal nodes of the tree and proper non-empty prefixes of $y$ is one-to-one (see picture below).

Labelling internal nodes with the $\prec$-ranks of their associated prefixes transforms the tree into a heap, i.e. ranks are increasing from leaves to the root. The relation between the infinite ordering and left Lyndon trees is established by the next result \cite{DolceRR19b}.

\begin{theorem}[Dolce, Restivo, Reutenauer, 2019]\label{theo-dolce}
For a Lyndon word $y$, the tree of internal nodes of the left Lyndon tree $\llynt(y)$ in which nodes are labelled by the ranks of proper non-empty prefixes of $y$ sorted according to the infinite ordering is the Cartesian tree of prefix ranks.
\end{theorem}

The picture below shows the left Lyndon tree of $\texttt{ababbababbabac}$ and the $\prec$-rank labels of its internal nodes.

\medskip
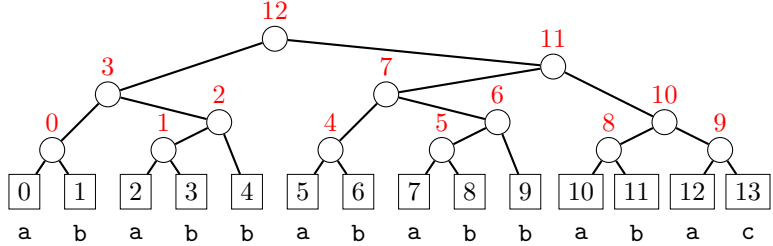
\tikzstyle{rank} = [circle, color=red, thin]
\begin{tikzpicture}[scale=0.37,node distance=1.5cm, auto,>=latex', thick]
\path[->]
 node[leaf] (p0) at (0,1.5) {$0$}
 node[leaf] (p1) at (2,1.5) {$1$}
 node[leaf] (p2) at (4,1.5) {$2$}
 node[leaf] (p3) at (6,1.5) {$3$}
 node[leaf] (p4) at (8,1.5) {$4$}
 node[leaf] (p5) at (10,1.5) {$5$}
 node[leaf] (p6) at (12,1.5) {$6$}
 node[leaf] (p7) at (14,1.5) {$7$}
 node[leaf] (p8) at (16,1.5) {$8$}
 node[leaf] (p9) at (18,1.5) {$9$}
 node[leaf] (p10) at (20,1.5) {$10$}
 node[leaf] (p11) at (22,1.5) {$11$}
 node[leaf] (p12) at (24,1.5) {$12$}
 node[leaf] (p13) at (26,1.5) {$13$}
;
\path[->]
 node (q0) at (0,0) {\tt a}
 node (q1) at (2,0) {\tt b}
 node (q2) at (4,0) {\tt a}
 node (q3) at (6,0) {\tt b}
 node (q4) at (8,0) {\tt b}
 node (q5) at (10,0) {\tt a}
 node (q6) at (12,0) {\tt b}
 node (q7) at (14,0) {\tt a}
 node (q8) at (16,0) {\tt b}
 node (q9) at (18,0) {\tt b}
 node (q10) at (20,0) {\tt a}
 node (q11) at (22,0) {\tt b}
 node (q12) at (24,0) {\tt a}
 node (q13) at (26,0) {\tt c}
;
\path[-]
 node[intern] (p14) at (1,3) {} node[rank] (p14l) at (1,4) {$0$}
 node[intern] (p15) at (5,3) {} node[rank] (p15l) at (5,4) {$1$}
 node[intern] (p16) at (7,4) {} node[rank] (p16l) at (7,5) {$2$}
 node[intern] (p17) at (3,5) {} node[rank] (p17l) at (3,6) {$3$}
 node[intern] (p18) at (11,3) {} node[rank] (p18l) at (11,4) {$4$}
 node[intern] (p19) at (15,3) {} node[rank] (p19l) at (15,4) {$5$}
 node[intern] (p20) at (17,4) {} node[rank] (p20l) at (17,5) {$6$}
 node[intern] (p21) at (13,5) {} node[rank] (p21l) at (13,6) {$7$}
 node[intern] (p22) at (21,3) {} node[rank] (p22l) at (21,4) {$8$}
 node[intern] (p23) at (25,3) {} node[rank] (p23l) at (25,4) {$9$}
 node[intern] (p24) at (23,4) {} node[rank] (p24l) at (23,5) {$10$}
 node[intern] (p25) at (19,6) {} node[rank] (p25l) at (19,7) {$11$}
 node[intern] (p26) at ( 9,7) {} node[rank] (p26l) at ( 9,8) {$12$}
 (p14) edge node {} (p0)
 (p14) edge node {} (p1)
 (p15) edge node {} (p2)
 (p15) edge node {} (p3)
 (p16) edge node {} (p4)
 (p16) edge node {} (p15)
 (p17) edge node {} (p14)
 (p17) edge node {} (p16)
 (p18) edge node {} (p5)
 (p18) edge node {} (p6)
 (p19) edge node {} (p7)
 (p19) edge node {} (p8)
 (p20) edge node {} (p9)
 (p20) edge node {} (p19)
 (p21) edge node {} (p18)
 (p21) edge node {} (p20)
 (p22) edge node {} (p10)
 (p22) edge node {} (p11)
 (p23) edge node {} (p12)
 (p23) edge node {} (p13)
 (p24) edge node {} (p22)
 (p24) edge node {} (p23)
 (p25) edge node {} (p21)
 (p25) edge node {} (p24)
 (p26) edge node {} (p17)
 (p26) edge node {} (p25)
;
\end{tikzpicture}

\medskip
Denoting a non-empty prefix of $y$ by the position of its last letter, the tables below show both $\prec$-ranks of proper non-empty prefixes of $y=\texttt{ababbababbabac}$ and its sorted list of prefixes, called the prefix standard permutation of $y$ in \cite{DolceRR19b}. They are denoted by $\ran$ and $\psp$ and are inverse of each other when considered as functions from $(0,1,\dots,|y|-2)$ to itself. The sorted list is $(0,2,3,1,5,7,8,6,10,12,11,9,4)$, that is, $\texttt{a}\prec\texttt{aba}\prec\texttt{abab}\prec\texttt{ab}
\prec\texttt{ababba}\prec\texttt{ababbaba}\prec\texttt{ababbabab}\prec\texttt{ababbab}
\prec\texttt{ababbababba}\prec\texttt{ababbababbaba}\prec\texttt{ababbababbab}
\prec\texttt{ababbababb}\prec\texttt{ababb}$.

\medskip\noindent{
\begin{tabular}{@{}l@{\quad}*{16}{p{4pt}}l@{}}
$j$		&&0&1&2&3&4&5&6&7&8&9&10&11&12&13 \\
\hline
$y[j]$	&&\tt a&\tt b&\tt a&\tt b&\tt b&\tt a&\tt b&\tt a&\tt b&\tt b&\tt a&\tt b&\tt a&\tt c \\
$\ran[j]$	&&0&3&1&2&12&4&7&5&6&11&8&10&9 \\[1mm]
$r$		&&0&1&2&3&4&5&6&7&8&9&10&11&12 \\
\hline
$\psp[r]$ &&0&2&3&1&5&7&8&6&10&12&11&9&4
\end{tabular}}

\medskip\noindent
The tree below is the Cartesian tree of prefix $\prec$-ranks.

\smallskip
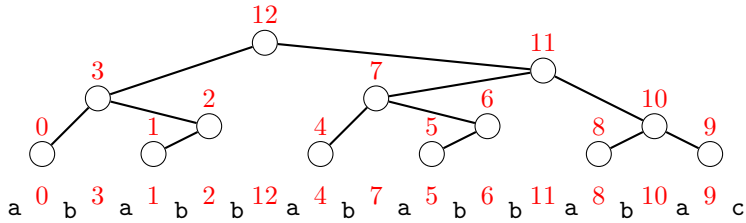
\begin{tikzpicture}[scale=0.37,node distance=1.5cm, auto,>=latex', thick]
\path[->]
 node[rank] (p0) at (1,1.5) {$0$}
 node[rank] (p1) at (3,1.5) {$3$}
 node[rank] (p2) at (5,1.5) {$1$}
 node[rank] (p3) at (7,1.5) {$2$}
 node[rank] (p4) at (9,1.5) {$12$}
 node[rank] (p5) at (11,1.5) {$4$}
 node[rank] (p6) at (13,1.5) {$7$}
 node[rank] (p7) at (15,1.5) {$5$}
 node[rank] (p8) at (17,1.5) {$6$}
 node[rank] (p9) at (19,1.5) {$11$}
 node[rank] (p10) at (21,1.5) {$8$}
 node[rank] (p11) at (23,1.5) {$10$}
 node[rank] (p12) at (25,1.5) {$9$}
;
\path[->]
 node (q0) at (0,1) {\tt a}
 node (q1) at (2,1) {\tt b}
 node (q2) at (4,1) {\tt a}
 node (q3) at (6,1) {\tt b}
 node (q4) at (8,1) {\tt b}
 node (q5) at (10,1) {\tt a}
 node (q6) at (12,1) {\tt b}
 node (q7) at (14,1) {\tt a}
 node (q8) at (16,1) {\tt b}
 node (q9) at (18,1) {\tt b}
 node (q10) at (20,1) {\tt a}
 node (q11) at (22,1) {\tt b}
 node (q12) at (24,1) {\tt a}
 node (q13) at (26,1) {\tt c}
;
\path[-]
 node[intern] (p14) at (1,3) {} node[rank] (p14l) at (1,4) {$0$}
 node[intern] (p15) at (5,3) {} node[rank] (p15l) at (5,4) {$1$}
 node[intern] (p16) at (7,4) {} node[rank] (p16l) at (7,5) {$2$}
 node[intern] (p17) at (3,5) {} node[rank] (p17l) at (3,6) {$3$}
 node[intern] (p18) at (11,3) {} node[rank] (p18l) at (11,4) {$4$}
 node[intern] (p19) at (15,3) {} node[rank] (p19l) at (15,4) {$5$}
 node[intern] (p20) at (17,4) {} node[rank] (p20l) at (17,5) {$6$}
 node[intern] (p21) at (13,5) {} node[rank] (p21l) at (13,6) {$7$}
 node[intern] (p22) at (21,3) {} node[rank] (p22l) at (21,4) {$8$}
 node[intern] (p23) at (25,3) {} node[rank] (p23l) at (25,4) {$9$}
 node[intern] (p24) at (23,4) {} node[rank] (p24l) at (23,5) {$10$}
 node[intern] (p25) at (19,6) {} node[rank] (p25l) at (19,7) {$11$}
 node[intern] (p26) at ( 9,7) {} node[rank] (p26l) at ( 9,8) {$12$}
 (p16) edge node {} (p15)
 (p17) edge node {} (p14)
 (p17) edge node {} (p16)
 (p20) edge node {} (p19)
 (p21) edge node {} (p18)
 (p21) edge node {} (p20)
 (p24) edge node {} (p22)
 (p24) edge node {} (p23)
 (p25) edge node {} (p21)
 (p25) edge node {} (p24)
 (p26) edge node {} (p17)
 (p26) edge node {} (p25)
;
\end{tikzpicture}

\medskip
The next theorem is the computational complement of Theorem~\ref{theo-dolce} showing additionally that the construction of the left Lyndon tree by Algorithm \Algo{LeftLyndonTree} processes the nodes of the tree in a left-to-right postorder traversal.

\begin{theorem}\label{theo-6}
Algorithm \Algo{LeftLyndonTree} applied to a Lyndon word $y$ of length $n>1$, creates and processes internal nodes of the tree $\llynt(y)$ in the order of their corresponding prefix ranks according to the infinite ordering $\prec$.
\end{theorem}

\begin{proof}
Since word $y$ is a Lyndon word not reduced to a single letter, it is of the form $x^kzb$ where $x$ is a Lyndon word of length $\period(x^kz)$, $k>0$, $z$ is a proper prefix of $x$ and $b$ is a letter greater than letter $a$ following prefix $z$ in $x$ (see \cite{Duval83}).

Algorithm \Algo{LeftLyndonTree} processes nodes of the tree $\llynt(y)$ as follows. First it builds $\llynt(x)$ and Lyndon subtrees of the next occurrences of $x$ in a left to right manner. It continues with the trees related to $z$. Eventually during the last bundling (run of instructions at lines~\ref{alg3-line10}-\ref{alg3-line15}) the algorithm builds $\llynt(zb)$ and follows with the nodes corresponding to the concatenations $x\cdot zb$, $x\cdot xzb$, \dots, $x\cdot x^{k-1}zb$ in that order.

We will prove the statement by induction on the length of the period $|x|$ of $x^kz$.
If $|x|=1$, $x$ is reduced to a single letter and $y$ is of the form $a^kb$ for two letters $a$ and $b$ with $a<b$. Nodes associated with prefixes $a^k$, $a^{k-1}$, \dots, $a$ are processed in this order, which matches the $\prec$-order of prefixes, $a^k \prec a^{k-1} \prec \cdots \prec a$, as expected.

We then assume $|x|>1$ and consider disjoint groups of non-empty proper prefixes of $y$. For $e=0,1,\dots,k$, let
$$P_e = \{x^eu \mbox{ prefix of } y \mid e|x| < |x^eu| < \min\{(e+1)|x|,|y|\}\}.$$
The main part of the proof relies on three claims that we prove first.

\paragraph{Claim 1:}
\textit{prefixes $x^eu\in P_e$, $0<e\leq k$, are in the same relative $\prec$-order as prefixes $u\in P_0$.}
Let $u,v\in P_0$ with $u \prec v$ and let us show $x^eu\prec x^ev$ considering two cases.

\medskip\noindent
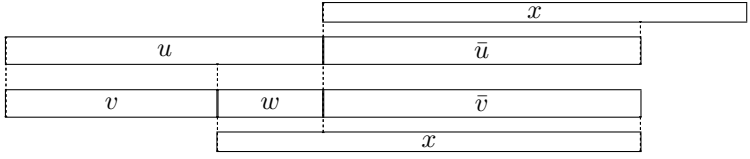
{\begin{picture}(300,60)(0,0)
\put(120,49){\framebox(160,7){$x$}} 
\put(  0,33){\framebox(120,10){$u$}} 
\put(120,33){\framebox(120,10){$\bar{u}$}} 
\put(240,33){\dashbox{1}(0,16)}
\put(  0,13){\framebox( 80,10){$v$}}
\put( 80,13){\framebox( 40,10){$w$}}
\put(120,13){\framebox(120,10){$\bar{v}$}}
\put(  0,13){\dashbox{1}(0,30)}
\put( 80,23){\dashbox{1}(0,10)}
\put( 80, 0){\framebox(160,7){$x$}} 
\put( 80, 0){\dashbox{1}(0,13)}
\put(120, 7){\dashbox{1}(0,49)}
\put(240, 0){\dashbox{1}(0,13)}
\end{picture}}

\smallskip
\textbf{Case $u^\infty = v^\infty$ and $|u|>|v|$.} By the Periodicity lemma $u$, $v$ and $v^{-1}u$ are powers of the same word. Let $w=v^{-1}u$, $\bar{v}=w^{-1}x$ and $\bar{u}$ the prefix of $x$ of length $|\bar{v}|$ (see picture). Since $x$ is a Lyndon word, $\bar{u} < x < \bar{v}$, which implies $ux < vx$ because $w$ is a prefix of $x$.
Therefore we have $(x^eu)^\infty < (x^ev)^\infty$, that is, $x^eu \prec x^ev$.

\textbf{Case $u^\infty < v^\infty$.}
Assume $u$ is shorter than $v$ and let $h$ be the largest exponent for which $u^h$ is a prefix of $v$. It is a proper prefix because $u^\infty < v^\infty$ and then $w=(u^h)^{-1}v$ is not empty.

If $|u|\leq|w|$, we have $u\strl w$, which implies $ux \strl vx$ and $(x^eu)^\infty < (x^ev)^\infty$, that is, $x^eu \prec x^ev$.

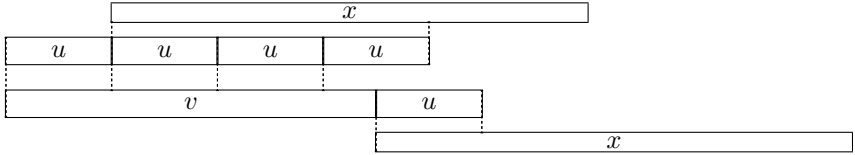
\medskip\noindent
{\begin{picture}(300,60)(0,0)
\put( 40,49){\framebox(180,7){$x$}} 

\put(  0,33){\framebox( 40,10){$u$}} 
\put( 40,33){\framebox( 40,10){$u$}} 
\put( 80,33){\framebox( 40,10){$u$}} 
\put(120,33){\framebox( 40,10){$u$}} 
\put(160,33){\dashbox{1}(0,16)}

\put(  0,13){\framebox(140,10){$v$}}
\put(140,13){\framebox( 40,10){$u$}}
\put(  0,13){\dashbox{1}(0,30)}
\put( 40,23){\dashbox{1}(0,26)}
\put( 80,23){\dashbox{1}(0,20)}
\put(120,23){\dashbox{1}(0,20)}

\put(140, 0){\framebox(180,7){$x$}} 
\put(140, 0){\dashbox{1}(0,23)}
\put(180, 7){\dashbox{1}(0,16)}
\end{picture}}

\smallskip
If $|u|>|w|$, $v$ is a proper prefix of $u^{h+1}$ but $u^{h+1}$ shorter than $vu$ cannot be a prefix of it due to the Periodicity lemma applied on periods $|u|$ and $|v|$ of $u^{h+1}$. Then $u\strl wu$ and since $u$ is a prefix of $x$ it implies $ux\strl vx$ and $(x^eu)^\infty < (x^ev)^\infty$, that is, $x^eu \prec x^ev$ as before.

The situation in which $u$ is longer than $v$ is fairly symmetric and treated similarly. Therefore again $u \prec v$ implies $x^eu \prec x^ev$, which proves the claim.

\paragraph{Claim 2:}
\textit{prefixes in $P_e$ are $\prec$-smaller than prefixes in $P_f$ when $0\leq e < f \leq k$.}
Let $u\in P_e$ and $v\in P_{f}$. We have to compare $u$ and $v$ according to $\prec$, that is, to compare $u^\infty$ and $v^\infty$.

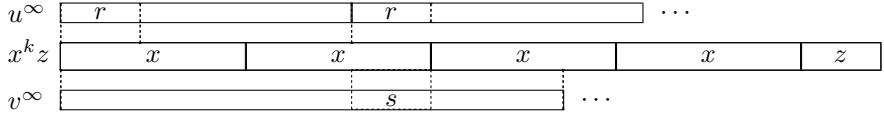
\medskip\noindent
{\begin{picture}(300,45)(0,0)
\put(  0,18){$x^kz$} 
\put( 20,15){\framebox(310,10){}}
\put( 20,15){\framebox( 70,10){$x$}}
\put( 90,15){\framebox( 70,10){$x$}}
\put(160,15){\framebox( 70,10){$x$}}
\put(230,15){\framebox( 70,10){$x$}}
\put(300,15){\framebox( 30,10){$z$}}
\put(  0,33){$u^\infty$}
\put( 20,33){\framebox(110,7){}} 
\put(130,33){\framebox(110,7){}} 
\put(240,33){\makebox( 25,7){$\cdots$}}
\put( 20,25){\dashbox{1}(0,15)}
\put( 20,33){\makebox(30,7){$r$}}
\put( 50,25){\dashbox{1}(0,15)}
\put(130,25){\dashbox{1}(0,15)}
\put(130,33){\makebox(30,7){$r$}}
\put(160,33){\dashbox{1}(0,7)}
\put(  0, 0){$v^\infty$}
\put( 20, 0){\framebox(190,7){}} 
\put(210, 0){\makebox( 25,7){$\cdots$}}
\put( 20, 0){\dashbox{1}(0,15)}
\put(210, 0){\dashbox{1}(0,15)}
\put(130, 0){\makebox(30,7){$s$}}
\put(130, 0){\dashbox{1}(30,15)}
\end{picture}}

\medskip
When $e>0$, $u$ is longer than $x$. Let $r$ be the prefix of $u$ for which $|ur|=|x^{e+1}|$ (see picture in which $u\in P_1$ and $v\in P_2$) and $s$ the suffix of $x$ of the same length.
Comparing $u^\infty$ and $v^\infty$ amounts to compare $r$ and $s$ because $u$ is a prefix of $v$. Since $r$ is a prefix and $s$ a suffix of the Lyndon word $x$, we have $r<s$ and even $r\strl s$, then $u^\infty < v^\infty$ and $u \prec v$.

When $e=0$, $u$ is shorter than $x$. Let then $h$ be the largest integer for which $u^h$ is a prefix of $x$. It is a proper prefix because $x$ is a Lyndon word and $w=(u^h)^{-1}x$ is not empty. As in the proof of previous claim, $u^{h+1}$ cannot be prefix of $xu$ that is a prefix of $v$. The same conclusion follows, that it, $u^{h+1}\strl vu$ and eventually $u\prec v$.

\paragraph{Claim 3:}
\textit{prefixes in $P_e$, $0\leq e\leq k$, are $\prec$-smaller than prefixes $x^f$, $0<f\leq k$.}
To prove the claim, in view of the statement of Claim 2 and the fact $x^k \prec x^{k-1} \prec x$ by definition, it is enough to show that $P_k \prec x^k$.
Note that if $P_k$ is empty the proof can be done with $P_{k-1}$ instead, and if in addition $k=1$ then we are left with an element in the proof of Claim 2.

Let $x^ku\in P_k$, $s=u^{-1}x$ and $r$ the prefix of $x$ of length $|s|$. As prefix and suffix of $x$, $r$ and $s$ satisfy $r < s$. Since $x^kur<x^kus=x^{k+1}$ and $r$ is a prefix of $x$, it results $(x^ku)^\infty < x^\infty$ and eventually $x^ku \prec x^k$. This prove the claim.

\smallskip
To summarise, claims show $$P_0 \prec P_1 \prec \cdots \prec P_k \prec x^k \prec x^{k-1} \prec \cdots \prec x.$$

Let us go back to induction. By induction hypothesis, the result holds for internal nodes of $\llynt(x)$ corresponding to prefixes in $P_0$.

Consider the next occurrences of $x$. Since the Lyndon suffix table for each of them is copied from that of prefix $x$ due to the instruction at line~\ref{alg3-line8} in Algorithm \Algo{LeftLyndonTree}, the Lyndon trees of all occurrences of $x$ have the same structure. Therefore, both from the induction hypothesis and from Claim 1, the order in which internal nodes of the $e$th occurrence of $x$ are processed and created matches the $\prec$-order of prefixes in $P_e$, for $0<e\leq k$.

The algorithm processes occurrences of $x$ from left to right, which corresponds to the result of Claim 2.
The treatment of $zb$ is done at the beginning of the bundling run, which also corresponds to the fact that prefixes in $P_k$ are $\prec$-larger than all prefixes that have been considered before.

Finally, the last part of the bundling creates nodes associated with $x^k$, $x^{k-1}$, \dots, $x$ in that order, which matches the order $x^k \prec x^{k-1} \prec \cdots \prec x$.

This ends the proof of the theorem.
\end{proof}

An immediate consequence of Theorem~\ref{theo-6} is that Algorithm \Algo{LeftLyndonTree} can be down-graded and adapted to compute directly the $\prec$-sorted list of non-empty proper prefixes of a Lyndon word, that is, to compute its prefix standard permutation (PSP). 
 See the details of this adaptation in the following algorithm.

\medskip\noindent
\begin{algo}{PrefixStandardPermutation}{y \textrm{ Lyndon word of length } n}
  \SET{\psp}{()}
  \SET{(\lyns[0],\per,i)}{(1,1,0)}
  \DOFORI{j}{1}{n-1}
    \IF{y[j] \neq y[i]} \RCOM{6}{$y[j] > y[i] = y[j-\per]$}
      \SET{\lyns[j]}{j+1}
      \SET{(\per,i)}{(j+1,0)}
    \ELSE
      \SET{\lyns[j]}{\lyns[i]}
      \SET{i}{i+1 \bmod \per}
    \FI
    \SET{(m,k)}{(1,j-1)}
    \DOWHILE{m < \lyns[j]}
      \SET{\psp}{\psp\cdot(j-m)}
      \SET{m}{m+\lyns[k]}
      \SET{k}{k-\lyns[k]}
    \OD
  \OD
  \RETURN{\psp}
\end{algo}

\begin{corollary}
Sorting the proper non-empty prefixes of a Lyndon word of length $n$ according to the infinite ordering $\prec$ can be done in time $O(n)$ in the letter-comparison model.
\end{corollary}

\begin{proof}
It essentially suffices to substitute the handling of sequence $\psp$ to the processing of internal nodes of the Lyndon tree in Algorithm \Algo{LeftLyndonTree}. The change is realised by Algorithm \Algo{PrefixStandardPermutation} above.
\end{proof}

\section{Reverse-engineering a PSP}\label{sect:reverse}

This section discusses how to recover a word of length $n$ from a permutation of $(0,1,\dots,n-2)$ assumed to be its prefix standard permutation (PSP).

We first consider the case of binary words on the alphabet $\{\texttt{a},\texttt{b}\}$. Function $\psp$ from $\mathcal{L}_n=\mathcal{L}\cap\{\texttt{a},\texttt{b}\}^n$ to the set of permutations of $(0,1,\cdots,n-2)$ is one-to-one. Thus $\psp^{-1}$ is a function from $\psp(\mathcal{L}_n)$ to $\mathcal{L}_n$
and $\psp^{-1}(\psp(y))=y$.  To show the property, given a permutation $p$ of $(0,1,\cdots,n-2)$, we propose the following algorithm to recover the possible word $y$ that admits the permutation as its PSP. 

\medskip\noindent
\begin{algo}{InversePsp}{p \textrm{ permutation of } (0,1,\dots,n-2)}
  \SET{\ran}{\mbox{inverse of } p}
  \SET{C}{\mbox{Cartesian tree of } \ran}
  \SET{T}{\mbox{$C$ extended with leaves to form a complete binary tree}}
  \SET{L}{\mbox{labelled $T$: left-child leaves labelled by \texttt{a} others by \texttt{b}}}
  \SET{y}{\mbox{word-label of leaves of } L}
  \IF{\psp(y)=p}
    \RETURN{y}
  \ELSE
    \RETURN{p\not\in \psp(\mathcal{L}_n)}
  \FI
\end{algo}

From the permutation $p=(1,0,4,3,5,2,6)$ the algorithm computes $\ran=(1,0,5,3,2,4,6)$ and eventually the labelled Lyndon tree below left. The word label of its leaves is $\texttt{aabaabbb}$ and effectively $\psp(\texttt{aabaabbb})=(1,0,4,3,5,2,6)$.

\medskip
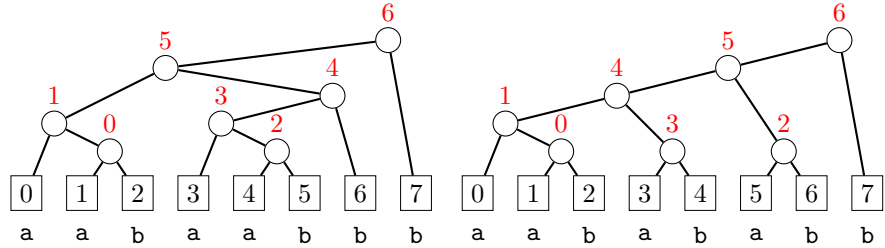
\begin{tikzpicture}[scale=0.37,node distance=1.5cm, auto,>=latex', thick]
\path[->]
 node[leaf] (p0) at (0,1.5) {$0$}
 node[leaf] (p1) at (2,1.5) {$1$}
 node[leaf] (p2) at (4,1.5) {$2$}
 node[leaf] (p3) at (6,1.5) {$3$}
 node[leaf] (p4) at (8,1.5) {$4$}
 node[leaf] (p5) at (10,1.5) {$5$}
 node[leaf] (p6) at (12,1.5) {$6$}
 node[leaf] (p7) at (14,1.5) {$7$}
;
\path[->]
 node (q0) at (0,0) {\tt a}
 node (q1) at (2,0) {\tt a}
 node (q2) at (4,0) {\tt b}
 node (q3) at (6,0) {\tt a}
 node (q4) at (8,0) {\tt a}
 node (q5) at (10,0) {\tt b}
 node (q6) at (12,0) {\tt b}
 node (q7) at (14,0) {\tt b}
;
\path[-]
 node[intern] (p8) at (3,3) {}   node (p8l) at (3,4) {\color{red} 0}
 node[intern] (p9) at (1,4) {}   node (p9l) at (1,5) {\color{red} 1}
 node[intern] (p10) at (9,3) {}  node (p10l) at (9,4) {\color{red} 2}
 node[intern] (p11) at (7,4) {}  node (p11l) at (7,5) {\color{red} 3}
 node[intern] (p12) at (11,5) {} node (p12l) at (11,6) {\color{red} 4}
 node[intern] (p13) at (5,6) {}  node (p13l) at (5,7) {\color{red} 5}
 node[intern] (p14) at (13,7) {} node (p14l) at (13,8) {\color{red} 6}
 (p8) edge node {} (p1)
 (p8) edge node {} (p2)
 (p9) edge node {} (p0)
 (p9) edge node {} (p8)
 (p10) edge node {} (p4)
 (p10) edge node {} (p5)
 (p11) edge node {} (p3)
 (p11) edge node {} (p10)
 (p12) edge node {} (p11)
 (p12) edge node {} (p6)
 (p13) edge node {} (p9)
 (p13) edge node {} (p12)
 (p14) edge node {} (p13)
 (p14) edge node {} (p7)
;
\end{tikzpicture}
\hfill
\begin{tikzpicture}[scale=0.37,node distance=1.5cm, auto,>=latex', thick]
\path[->]
 node[leaf] (p0) at (0,1.5) {$0$}
 node[leaf] (p1) at (2,1.5) {$1$}
 node[leaf] (p2) at (4,1.5) {$2$}
 node[leaf] (p3) at (6,1.5) {$3$}
 node[leaf] (p4) at (8,1.5) {$4$}
 node[leaf] (p5) at (10,1.5) {$5$}
 node[leaf] (p6) at (12,1.5) {$6$}
 node[leaf] (p7) at (14,1.5) {$7$}
;
\path[->]
 node (q0) at (0,0) {\tt a}
 node (q1) at (2,0) {\tt a}
 node (q2) at (4,0) {\tt b}
 node (q3) at (6,0) {\tt a}
 node (q4) at (8,0) {\tt b}
 node (q5) at (10,0) {\tt a}
 node (q6) at (12,0) {\tt b}
 node (q7) at (14,0) {\tt b}
;
\path[-]
 node[intern] (p8) at (3,3) {}   node (p8l) at (3,4) {\color{red} 0}
 node[intern] (p9) at (1,4) {}   node (p9l) at (1,5) {\color{red} 1}
 node[intern] (p10) at (7,3) {}  node (p10l) at (7,4) {\color{red} 3}
 node[intern] (p11) at (5,5) {}  node (p11l) at (5,6) {\color{red} 4}
 node[intern] (p12) at (11,3) {} node (p12l) at (11,4) {\color{red} 2}
 node[intern] (p13) at (9,6) {}  node (p13l) at (9,7) {\color{red} 5}
 node[intern] (p14) at (13,7) {} node (p14l) at (13,8) {\color{red} 6}
 (p8) edge node {} (p1)
 (p8) edge node {} (p2)
 (p9) edge node {} (p0)
 (p9) edge node {} (p8)
 (p10) edge node {} (p3)
 (p10) edge node {} (p4)
 (p11) edge node {} (p9)
 (p11) edge node {} (p10)
 (p12) edge node {} (p5)
 (p12) edge node {} (p6)
 (p13) edge node {} (p11)
 (p13) edge node {} (p12)
 (p14) edge node {} (p13)
 (p14) edge node {} (p7)
;
\end{tikzpicture}

However with the permutation $p=(1,0,5,3,2,4,6)$, the algorithm computes $\ran=(1,0,4,3,5,2,6)$ and the correponding $L$ tree (above right), which produces the word $\texttt{aabababb}$. But $\psp(\texttt{aabababb})=(1,0,3,2,5,4,6)$ is not the input permutation. This is because obviously not all the $(n-1)!$ permutations are PSPs of some binary Lyndon words (less than $2^n$). It may also happen that word $y$ built in the procedure is not even a Lyndon word.

\begin{proposition}
On a binary alphabet the prefix standard permutation is a one-to-one function and computing the Lyndon word $y$ for which $\psp(y)$ is a given valid permutation can be done in linear time.
\end{proposition}

\begin{proof}
From the above discussion and Algorithm \Algo{InversePsp}, the one-to-one feature is a consequence of Theorem~\ref{theo-dolce}. As for the running time it comes from the linearity of all operations, especially those of the Cartesian tree construction\footnote{See for example \texttt{https://en.wikipedia.org/wiki/Cartesian\_tree}} and of the prefix standard permutation computation by Algorithm \Algo{PrefixStandardPermutation} in Section~\ref{sect:sort}.
\end{proof}

On alphabets with more than two letters the function $\psp$ is not one-to-one. For example $(0, 2, 3, 1, 4)$ is the PSP of Lyndon words \texttt{ababbb}, \texttt{ababbc}, \texttt{ababcb} and \texttt{ababcc}, and permutation $(0,1,2,3)$ is the PSP of all (Lyndon) words in $\texttt{a}\{\texttt{b},\texttt{c}\}^4$.


Nevertheless, given the permutation $p=\psp(z)$ associated with a Lyndon word $z$ of length $n$, we can compute an equivalent word $y$ whose PSP is $p$. The basic element to do it is to deal with prefix periods of the word.

Indeed, periods of prefixes of $y$ can be retrieved from $p$ by looking at some positions where $p$ is decreasing. Due to properties (proof of theorem 6, after claim 3; the only case where a longer prefix is $\prec$-smaller than a shorter prefix is when the shorter one is a period of the longer) and the definition of the prefix standard permutation, when there is a decrease in the order of prefixes it is because there is a non-empty border.  Therefore scanning $p$ from right to left enables tracing the periodicity of each proper prefix. This is how Algorithm \Algo{PeriodsFromPsp} computes the period table of a word from its PSP.


\begin{algo}{PeriodsFromPsp}{p \mbox{ PSP of a Lyndon word of length } n}
  \SET{q}{n}
  \DOFORD{j}{n-2}{1}
    \IF{j \geq q}
      \SET{\per[j]}{q}
    \ELSEIF{p[j] < p[j-1]} \label{algo:pfp-5}
      \SET{\per[j]}{q \leftarrow p[j]+1}
    \ELSE
      \SET{\per[j]}{j+1}
    \FI
  \OD
  \SET{\per[0]}{1}
  \RETURN{\per}
\end{algo}


In the example below, positions on the PSP $p$, where condition at line~\ref{algo:pfp-5} is met, are $j=6$ and $j=2$ corresponding respectively to periods $4$ and $2$.

Here is the step-by-step computation of the periods.
For the following example, we start at $j=7$, since $p[j]>p[j-1]$ then $\per[j]= j+1 =8$ , now $p[6]<p[5]$ which means $\per[6]= p[6]+1= 4$. Next, we can move on to position $p[6]-1$;  $p[2]< p[1]$ so $\per[2]= 1+1$, and we are done.

\medskip
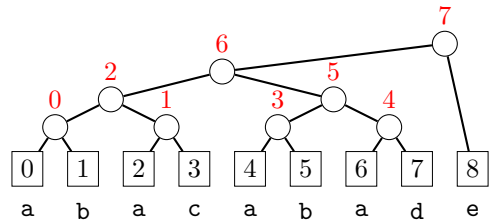
\begin{tikzpicture}[scale=0.37,node distance=1.5cm, auto,>=latex', thick]
\path[->]
 node[leaf] (p0) at (0,1.5) {$0$}
 node[leaf] (p1) at (2,1.5) {$1$}
 node[leaf] (p2) at (4,1.5) {$2$}
 node[leaf] (p3) at (6,1.5) {$3$}
 node[leaf] (p4) at (8,1.5) {$4$}
 node[leaf] (p5) at (10,1.5) {$5$}
 node[leaf] (p6) at (12,1.5) {$6$}
 node[leaf] (p7) at (14,1.5) {$7$}
 node[leaf] (p8) at (16,1.5) {$8$}
;
\path[->]
 node (q0) at (0,0) {\tt a}
 node (q1) at (2,0) {\tt b}
 node (q2) at (4,0) {\tt a}
 node (q3) at (6,0) {\tt c}
 node (q4) at (8,0) {\tt a}
 node (q5) at (10,0) {\tt b}
 node (q6) at (12,0) {\tt a}
 node (q7) at (14,0) {\tt d}
 node (q8) at (16,0) {\tt e}
;
\path[-]
 node[intern] (p9) at (1,3) {}   node (p9l) at (1,4) {\color{red} 0}
 node[intern] (p10) at (5,3) {}   node (p10l) at (5,4) {\color{red} 1}
 node[intern] (p11) at (3,4) {}  node (p11l) at (3,5) {\color{red} 2}
 node[intern] (p12) at (9,3) {}  node (p12l) at (9,4) {\color{red} 3}
 node[intern] (p13) at (13,3) {} node (p13l) at (13,4) {\color{red} 4}
 node[intern] (p14) at (11,4) {} node (p14l) at (11,5) {\color{red} 5}
 node[intern] (p15) at (7,5) {}  node (p15l) at (7,6) {\color{red} 6}
 node[intern] (p16) at (15,6) {} node (p16l) at (15,7) {\color{red} 7}

 (p9) edge node {} (p0)
 (p9) edge node {} (p1)
 (p10) edge node {} (p2)
 (p10) edge node {} (p3)
 (p11) edge node {} (p9)
 (p11) edge node {} (p10)
 (p12) edge node {} (p4)
 (p12) edge node {} (p5)
 (p13) edge node {} (p6)
 (p13) edge node {} (p7)
 (p14) edge node {} (p12)
 (p14) edge node {} (p13)
 (p15) edge node {} (p11)
 (p15) edge node {} (p14)
 (p16) edge node {} (p15)
 (p16) edge node {} (p8)
;
\end{tikzpicture}

\medskip
\begin{tabular}{@{}l@{\quad}*{13}{p{9pt}}l@{}}
$j$		&0&1&2&3&4&5&6&7&8 \\ 
\hline
$y[j]$	&\tt a&\tt b&\tt a&\tt c&\tt a&\tt b&\tt a&\tt d&\tt e\\
$\psp[j]$    &0&2&1&4&6&5&3&7\\
$\ran[j]$    &0&2&1&6&3&5&4&7\\
$\per[j]$    &1&2&2&4&4&4&4&8&9\\
\end{tabular}

\bigskip
Another way to retrieve periods of prefixes is to look at prefix ranks according to the infinite order. To do so, it amounts to look at ranks of proper Lyndon prefixes of $y$, because their periods are their lengths, starting with the first rank, $r$. Then the next Lyndon prefix is the shortest prefix having a rank greater than $r$, which is iterated until the end. This amounts to go up the left Lyndon tree from its leftest leaf to its root. In the example (above) positions on the rank table that correspond to the traversal are $j=0,1,3,7$.

%
%

Following the discussion, Algorithm \Algo{WordFromPsp} takes as input the PSP $p$ of a Lyndon word and builds an equivalent word, that is, a Lyndon word having the same PSP. The output is a word on the (constant) alphabet $\{\texttt{a},\texttt{b},\cdots\}$. If $y\in\mathcal{L}_2$, the output is $y$ itself. Else, the output is the smallest lexicographic Lyndon word having the same PSP.

After the inversion of $p$ to get the table $\ran$ (lines~\ref{algo-wfp-1}-\ref{algo-wfp-2}), the algorithm proceeds online on that table. It keeps information on the last highest rank met so far in variable $r$ and on the current period of the being-built word $y$ in variable $q$. Instruction at lines~\ref{algo-wfp-5}-\ref{algo-wfp-8} implements the bottom up description on the virtual left Lyndon tree of the future output.

\medskip\noindent
\begin{algo}{WordFromPsp}{p \mbox{ PSP of a Lyndon word of length } n}
  \DOFORI{j}{0}{n-1}        \label{algo-wfp-1}
    \SET{\ran[p[j]]}{j}   \label{algo-wfp-2}
  \OD
  \SET{(y,r,q)}{(\texttt{a},\ran[0],1)}
  \DOFORI{j}{0}{n-2}
    \IF{\ran[j] \leq r}       \label{algo-wfp-5}
      \SET{y}{y\cdot y[j-q]}
    \ELSE
      \SET{y}{y\cdot (\mbox{smallest letter larger than } y[j-q])}
      \SET{(r,q)}{(\ran[j],j+1)}    \label{algo-wfp-8}
    \FI
  \OD
  \SET{y}{y \cdot (\mbox{smallest letter larger than } y[n-1-q])}
  \RETURN{y}
\end{algo}

Applied to the example whose PSP is $(0,2,1,4,6,5,3,7)=\psp(\texttt{abacabade})$ the algorithm produces the Lyndon word $\texttt{abacabadb}$. Indeed, in the initial word, letter $\texttt{b}$ is necessarily greater than $\texttt{a}$, letter $\texttt{c}$ greater than $\texttt{b}$ and letter $\texttt{d}$ greater than $\texttt{c}$. But letter $\texttt{e}$ is only required to be greater than $\texttt{a}$.

\begin{proposition}
Given the PSP table $p$ of a Lyndon word, \Call{WordFromPsp}{p} is the lexicographic smallest Lyndon word $y\in\{\texttt{a},\texttt{b},\cdots\}$ for which $\psp(y)=p$. The computation is done in linear time.
\end{proposition}

Note that when applied to the PSP of a half Zimin word the algorithm recovers the word itself up to an alphabetic translation. Recall that Zimin words $Z_i$ are defined by the relations: $Z_0$ is the empty word and, for $i>0$, $Z_i=Z_{i-1}\cdot a_i\cdot Z_{i-1}$, where $a_i$ is a letter not occurring in $Z_{i-1}$. Using the constant alphabet, first half Zimin words are $\epsilon$, $\texttt{a}$, $\texttt{ab}$, $\texttt{abac}$, $\texttt{abacabad}$ and $\texttt{abacabadabacabae}$.

Half Zimin words contain the largest alphabet amongst the class of solution
words of length $n$ constructed by Algorithm \Algo{WordFromPsp}.
They contain $\lfloor\log (n+1)\rfloor +1$ distinct letters.

%

\section{Lyndon forest}\label{sect:forest}

Methods of previous sections that concern Lyndon words easily extend to all (non-empty) words. Trees become forests due to the Lyndon factorisation of words. A forest is reduced to a single tree when the considered word is a Lyndon word.

The Lyndon factorisation of a non-empty word $y$ is a decreasing list of Lyndon factors of the word. It is a list $x_1, x_2, \dots, x_k$ for which both $x_1x_2\cdots x_k=y$ and $x_1\geq x_2\geq\cdots \geq x_k$ hold. This factorisation is unique (see \cite[Theorem 5.1.5]{Lothaire83}) and
the left Lyndon forest of word $y$ is the list of left Lyndon trees $\llynt(x_1)$, $\llynt(x_2)$, \dots, $\llynt(x_k)$.

The factorisation and its algorithm by Duval \cite{Duval83} is the guiding thread of previous algorithms. Following the techniques in Section~\ref{sect:llyn} the computation of Lyndon forest also uses the Lyndon suffix table of the word.
Algorithm \Algo{LyndonSuffixTable} deals with words that are not necessarily Lyndon words, and it can be viewed as an extension of Algorithm \Algo{LyndonSuffixT}.

Computing the forest from the table can then be carried out as in Section~\ref{sect:llyn}, therefore we only describe the table computation below.

\medskip\noindent
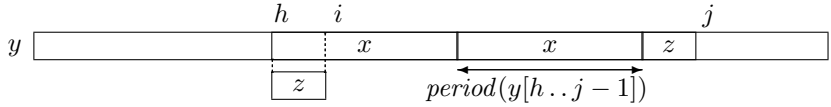
{\begin{picture}(300,37)(0,0)
\put(  0,18){$y$}
\put( 10,15){\framebox(300,10){}}
\put(100,15){\framebox( 70,10){$x$}}
\put(170,15){\framebox( 70,10){$x$}}
\put(240,15){\framebox(20,10){$z$}}
\put(100,29){\makebox(7,7){$h$}}
\put(120,29){\makebox(10,7){$i$}}
\put(260,28){\makebox(10,7){$j$}}
\put(100, 0){\framebox(20,10){$z$}}
\put(170,11){\vector(1,0){70}}
\put(240,11){\vector(-1,0){70}}
\put(170, 0){\makebox( 60,7){$\period(y[h\dd j-1])$}}
\put(100,10){\dashbox{1}(0,5)}
\put(120,10){\dashbox{1}(0,15)}
\end{picture}}

\medskip\noindent
\begin{algo}{LyndonSuffixTable}{y \textrm{ non-empty word of length } n}
  \SET{\lyns[0]}{1}
  \SET{(\per,h,i,j)}{(1,0,0,1)}
  \DOWHILE{j < n}
    \IF{y[j]<y[i]}                \label{alg5-line-4}
      \SET{h}{j-(i-h)}
      \SET{\lyns[h]}{1}
      \SET{(\per,i,j)}{(1,h,h+1)} \label{alg5-line-7}
    \ELSEIF{y[j] > y[i]}
      \SET{\lyns[j]}{j-h+1}
      \SET{j}{j+1}                \label{alg5-line-10}
      \SET{(\per,i)}{(j-h,h)}
    \ELSE
      \SET{\lyns[j]}{\lyns[i]}
      \SET{(i,j)}{(h+(i-h+1 \bmod \per),j+1)} \label{alg5-line-13}
    \FI
  \OD
  \RETURN{\lyns}
\end{algo}

The update of Algorithm \Algo{LyndonSuffixT} to get Algorithm \Algo{LyndonSuffixTable} essentially lies in instructions on lines~\ref{alg5-line-4}-\ref{alg5-line-7} in the latter algorithm above. They reset the computation to the suffix $y[h\dd n-1]$ of the input after the factorisation of the prefix $y[0\dd h-1]$ is definitely achieved. Variable $h$ becomes the starting position of the next Lyndon factor of $y$.

\begin{proposition}\label{prop-5}
Algorithm \Algo{LyndonSuffixTable} computes the Lyndon suffix table of a word of length $n>0$ in time $O(n)$ in the letter-comparison model.
\end{proposition}

\begin{proof}
Let us consider the values of expression $h+j$ and show they strictly increase after each iteration of the \While\ loop. The claim holds if the condition at line~\ref{alg5-line-4} is false, because $j$ is incremented by at least one unit (on line~\ref{alg5-line-10} or on line~\ref{alg5-line-13}) and $h$ remains unchanged. The claim also holds if the condition at line~\ref{alg5-line-4} is true, because $h$ is incremented by at least $\period(y[h\dd j-1])$ while $j$ is decremented by less than the same value.

Thus, since $h+j$ goes from $1$ to at most $2n-1$ twinned with the fact that instruction at lines~\ref{alg5-line-4}-\ref{alg5-line-13} executes in constant time, the running time is $O(n)$.
\end{proof}

Note that the Lyndon factorisation of a word $y$ can be retrieved from its $\lyns$ table by sequentially tracing back from $|y|$ starting positions of previous factors. The list of starting positions of factors, in reverse order, is $i_k=|y|-\lyns[|y|-1]$, $i_{k-1}=i_k-\lyns[i_{k-1}-1], \dots, 0$.

The Lyndon suffix table of $y = \texttt{babbababbaabb}$ is as follows:

\medskip\noindent
\begin{tabular}{@{}l@{\quad}*{13}{p{9pt}}l@{}}
$j$		&0&1&2&3&4&5&6&7&8&9&10&11&12 \\
\hline
$y[j]$	&\tt b&\tt a&\tt b&\tt b&\tt a&\tt b&\tt a&\tt b&\tt b&\tt a&\tt a&\tt b&\tt b \\
$\lyns[j]$	&1&1&2&3&1&2&1&2&5&1&1&3&4
\end{tabular}

\medskip\noindent
Starting positions of factors of its Lyndon factorisation are $9=13-\lyns[12]$, $4=9-\lyns[8]$, $1=4-\lyns[3]$, $0=1-\lyns[0]$.
The bellow figure depicts the Lyndon forest of this example.

\medskip
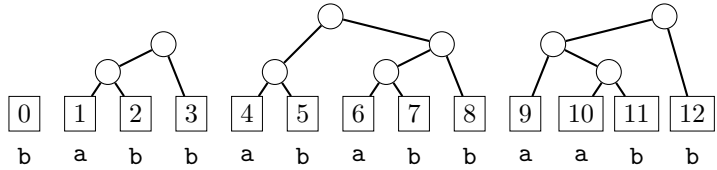
\begin{tikzpicture}[scale=0.37,node distance=1.5cm, auto,>=latex', thick]
\path[->]
 node[leaf] (p0) at (0,1.5) {$0$}
 node[leaf] (p1) at (2,1.5) {$1$}
 node[leaf] (p2) at (4,1.5) {$2$}
 node[leaf] (p3) at (6,1.5) {$3$}
 node[leaf] (p4) at (8,1.5) {$4$}
 node[leaf] (p5) at (10,1.5) {$5$}
 node[leaf] (p6) at (12,1.5) {$6$}
 node[leaf] (p7) at (14,1.5) {$7$}
 node[leaf] (p8) at (16,1.5) {$8$}
 node[leaf] (p9) at (18,1.5) {$9$}
 node[leaf] (p10) at (20,1.5) {$10$}
 node[leaf] (p11) at (22,1.5) {$11$}
 node[leaf] (p12) at (24,1.5) {$12$}
;
\path[->]
 node (q0) at (0,0) {\tt b}
 node (q1) at (2,0) {\tt a}
 node (q2) at (4,0) {\tt b}
 node (q3) at (6,0) {\tt b}
 node (q4) at (8,0) {\tt a}
 node (q5) at (10,0) {\tt b}
 node (q6) at (12,0) {\tt a}
 node (q7) at (14,0) {\tt b}
 node (q8) at (16,0) {\tt b}
 node (q9) at (18,0) {\tt a}
 node (q10) at (20,0) {\tt a}
 node (q11) at (22,0) {\tt b}
 node (q12) at (24,0) {\tt b}
;
\path[-]
 node[intern] (p14) at (3,3) {}
 node[intern] (p15) at (5,4) {}
 node[intern] (p16) at (9,3) {}
 node[intern] (p17) at (13,3) {}
 node[intern] (p18) at (15,4) {}
 node[intern] (p19) at (11,5) {}
 node[intern] (p20) at (21,3) {}
 node[intern] (p21) at (19,4) {}
 node[intern] (p22) at (23,5) {}
 (p14) edge node {} (p1)
 (p14) edge node {} (p2)
 (p15) edge node {} (p14)
 (p15) edge node {} (p3)
 (p16) edge node {} (p4)
 (p16) edge node {} (p5)
 (p17) edge node {} (p6)
 (p17) edge node {} (p7)
 (p18) edge node {} (p17)
 (p18) edge node {} (p8)
 (p19) edge node {} (p16)
 (p19) edge node {} (p18)
 (p20) edge node {} (p10)
 (p20) edge node {} (p11)
 (p21) edge node {} (p9)
 (p21) edge node {} (p20)
 (p22) edge node {} (p21)
 (p22) edge node {} (p12)
;
\end{tikzpicture}

\medskip
Algorithm \Algo{LeftLyndonForest} is merely adapted from the previous algorithm in order to manage Lyndon tree constructions of each factor of the Lyndon factorisation while computing the latter. The next proposition is a direct consequence of Proposition~\ref{prop-5}.

\begin{proposition}\label{prop-6}
Algorithm \Algo{LeftLyndonForest} computes the Lyndon forest of a word of length $n>0$ in time $O(n)$ in the letter-comparison model.
\end{proposition}

\medskip\noindent
\begin{algo}{LeftLyndonForest}{y \textrm{ non-empty word of length } n}
  \SET{(\lyns[0],\roo[0])}{(1,0)}
  \SET{(\per,h,i,j)}{(1,0,0,1)}
  \DOWHILE{j < n}
    \SET{\roo[j]}{j}
    \IF{y[j]<y[i]}
      \SET{h}{j-(i-h)}
      \SET{\lyns[h]}{1}
      \SET{(\per,i,j)}{(1,h,h+1)}
    \ELSEIF{y[j] > y[i]}
      \SET{\lyns[j]}{j-h+1}
      \SET{j}{j+1}
      \SET{(\per,i)}{(j-h,h)}
    \ELSE
      \SET{\lyns[j]}{\lyns[i]}
      \SET{(i,j)}{(h+(i-h+1 \bmod \per),j+1)}
    \FI
    \COM{Bundle}
    \SET{(p,m,k)}{(\roo[j],1,j-1)}
    \DOWHILE{m < \lyns[j]}
      \SET{q}{\mbox{new node} \geq n}
      \SET{(\leftchild[q],\rightchild[q])}{(\roo[k],p)}
      \SET{(p,m)}{(q,m+\lyns[k])}
      \SET{k}{k-\lyns[k]}
    \OD
  \OD
  \RETURN{\roo[n-1]}
\end{algo}

\section{Conclusions}
In this paper, Algorithm \Algo{LyndonSuffixTable} computes the Lyndon suffix table of a word. The table is an essential part of algorithm  \Algo{LeftLyndonTree} that constructs the left Lyndon tree of a Lyndon word in linear time.

We further investigated the prefix standard permutation of a Lyndon word, initially introduced by Dolce et al. \cite{DolceRR19b}, and its relation to the left Lyndon tree. This study resulted in a linear-time algorithm for sorting the prefixes of a Lyndon word according to infinite ordering. In addition, we showed how to recover a word from a given permutation assumed to be a prefix standard permutation.

To achieve the results, we exhibited a strong connection between the prefix ranks and the left Lyndon tree. This connection dictates that the order in which the internal nodes of the left Lyndon tree are created and processed coincides with that of the prefix ranks according to infinite ordering and corresponds to the left-to-right postorder traversal of the tree.

We finally endeavoured to design a linear-time algorithm, \Algo{LeftLyndonForest}, that computes the Lyndon forest of an ordinary word.



Many interesting questions remain, among them are: Is there a connection between runs and the internal nodes of the left Lyndon forest? Is there a tight relation between the left Lyndon trees and the right Lyndon trees?

\bibliographystyle{abbrv}
\bibliography{LeftLyndonTree}

\begin{thebibliography}{10}

\bibitem{BadkobehC20}
G.~Badkobeh and M.~Crochemore.
\newblock Left {L}yndon tree construction.
\newblock In J.~Holub and J.~Zd{\'{a}}rek, editors, {\em Prague Stringology
  Conference 2020, Prague, Czech Republic, August 31-September 2, 2020}, pages
  84--95. Czech Technical University in Prague, Faculty of Information
  Technology, Department of Theoretical Computer Science, 2020.

\bibitem{BannaiIINTT17}
H.~Bannai, T.~I, S.~Inenaga, Y.~Nakashima, M.~Takeda, and K.~Tsuruta.
\newblock The ``runs'' theorem.
\newblock {\em {SIAM} J. Comput.}, 46(5):1501--1514, 2017.

\bibitem{BilleE0GKMR20}
P.~Bille, J.~Ellert, J.~Fischer, I.~L. G{\o}rtz, F.~Kurpicz, J.~I. Munro, and
  E.~Rotenberg.
\newblock Space efficient construction of {L}yndon arrays in linear time.
\newblock In A.~Czumaj, A.~Dawar, and E.~Merelli, editors, {\em 47th
  International Colloquium on Automata, Languages, and Programming, {ICALP}
  2020, July 8-11, 2020, Saarbr{\"{u}}cken, Germany (Virtual Conference)},
  volume 168 of {\em LIPIcs}, pages 14:1--14:18. Schloss Dagstuhl -
  Leibniz-Zentrum f{\"{u}}r Informatik, 2020.

\bibitem{CrochemoreIKRRW12}
M.~Crochemore, C.~S. Iliopoulos, M.~Kubica, J.~Radoszewski, W.~Rytter, and
  T.~Walen.
\newblock The maximal number of cubic runs in a word.
\newblock {\em J. Comput. Syst. Sci.}, 78(6):1828--1836, 2012.

\bibitem{CLR20cup}
M.~Crochemore, T.~Lecroq, and W.~Rytter.
\newblock {\em 125 Problems in Text Algorithms}.
\newblock Cambridge University Press, 2021.
\newblock In press.

\bibitem{CrochemoreR20}
M.~Crochemore and L.~M.~S. Russo.
\newblock Cartesian and {L}yndon trees.
\newblock {\em Theoretical Computer Science}, 806:1--9, February 2020.

\bibitem{DolceRR19}
F.~Dolce, A.~Restivo, and C.~Reutenauer.
\newblock On generalized {L}yndon words.
\newblock {\em Theor. Comput. Sci.}, 777:232--242, 2019.

\bibitem{DolceRR19b}
F.~Dolce, A.~Restivo, and C.~Reutenauer.
\newblock Some variations on {L}yndon words.
\newblock {\em CoRR}, abs/1904.00954, 2019.

\bibitem{Duval83}
J.~Duval.
\newblock Factorizing words over an ordered alphabet.
\newblock {\em J. Algorithms}, 4(4):363--381, 1983.

\bibitem{FranekL19}
F.~Franek and M.~Liut.
\newblock Algorithms to compute the {L}yndon array revisited.
\newblock In J.~Holub and J.~Zd{\'{a}}rek, editors, {\em Prague Stringology
  Conference 2019, Prague, Czech Republic, August 26-28, 2019}, pages 16--28.
  Czech Technical University in Prague, Faculty of Information Technology,
  Department of Theoretical Computer Science, 2019.

\bibitem{HohlwegR03}
C.~Hohlweg and C.~Reutenauer.
\newblock Lyndon words, permutations and trees.
\newblock {\em Theor. Comput. Sci.}, 307(1):173--178, 2003.

\bibitem{KolpakovK99}
R.~M. Kolpakov and G.~Kucherov.
\newblock Finding maximal repetitions in a word in linear time.
\newblock In {\em 40th Annual Symposium on Foundations of Computer Science,
  {FOCS} '99, 17-18 October, 1999, New York, NY, {USA}}, pages 596--604. {IEEE}
  Computer Society, 1999.

\bibitem{Lothaire83}
M.~Lothaire.
\newblock {\em Combinatorics on Words}.
\newblock Addison-Wesley, 1983.
\newblock Reprinted in 1997.

\bibitem{Lyndon54}
R.~C. Lyndon.
\newblock On {B}urnside problem {I}.
\newblock {\em Trans. Amer. Math. Soc.}, 77:202--215, 1954.

\bibitem{ManberM90}
U.~Manber and G.~Myers.
\newblock Suffix arrays: {A} new method for on-line string searches.
\newblock In D.~S. Johnson, editor, {\em Proceedings of the First Annual
  {ACM-SIAM} Symposium on Discrete Algorithms, 22-24 January 1990, San
  Francisco, California, {USA}}, pages 319--327. {SIAM}, 1990.

\bibitem{NakashimaTIBT19}
Y.~Nakashima, T.~Takagi, S.~Inenaga, H.~Bannai, and M.~Takeda.
\newblock On the size of the smallest alphabet for {L}yndon trees.
\newblock {\em Theor. Comput. Sci.}, 792:131--143, 2019.

\bibitem{Ufnarovskij11}
V.~A. Ufnarovskij.
\newblock Combinatorial and asymptotic methods in algebra.
\newblock In A.~Kostrikin and I.~Shafarevich, editors, {\em Algebra VI:
  Combinatorial and Asymptotic Methods of Algebra. Non-Associative Structures},
  volume~57 of {\em Encyclopaedia of Mathematical Sciences}, pages 1--196.
  Springer, Berlin, 2011.

\bibitem{Viennot78}
G.~Viennot.
\newblock {\em Alg\`ebres de Lie libres et mono\"ides libres}, volume 691 of
  {\em Lecture Notes in Mathematics}.
\newblock Springer-Verlag, Berlin, 1978.

\end{thebibliography}
\end{document}